\documentclass[11pt]{llncs}

\usepackage{graphicx}
\usepackage{llncsdoc}
\usepackage{fullpage}

\begin{document}

\title{Mod/Resc Parsimony Inference}

\author{
Igor Nor \inst{1,2,*}
\and%
Danny Hermelin \inst{3}
\and%
Sylvain Charlat \inst{1}
\and%
Jan Engelstadter \inst{4}
\and%
Max Reuter \inst{5}
\and%
Olivier Duron \inst{6}
\and%
Marie-France Sagot \inst{1,2,*}}

\institute{
Universit\'e de Lyon, F-69000, Lyon ; Universit\'e Lyon 1 ;
  CNRS, UMR5558
\and
Bamboo Team, INRIA Grenoble Rh\^one-Alpes, France
\and
Max Planck Institute for Informatics, Saarbr\"{u}cken - Germany
\and
Institute of Integrative Biology, ETH Zurich, Switzerland
\and
University College London, UK
\and
Institute of Evolutionary Sciences; CNRS - University of Montpellier II, France
\\
$^*$
{\it Corresponding authors:} \email{norigor@gmail.com,
Marie-France.Sagot@inria.fr}}

% Not principal: in the corresponding email, it looks like Danny is not happy he is not in the list. From my point of view, he helped a lot.

\date{}

\maketitle

\vspace{-0.8cm}
\begin{abstract}
We address in this paper a new computational biology problem that aims
at understanding a mechanism that could potentially be used
to genetically manipulate natural insect populations
infected by inherited, intra-cellular parasitic bacteria.
In this
problem, that we denote by \textsc{Mod/Resc
  Parsimony Inference}, we are given a boolean matrix and the goal is to find
two other boolean matrices with a minimum number of columns such that
an appropriately defined operation on these matrices gives back the
 input. We show that this is formally equivalent to the
\textsc{Bipartite Biclique Edge Cover} problem and
derive some complexity results for our problem using this
equivalence. We provide a new,
fixed-parameter tractability approach for solving both
that slightly improves upon a
previously published algorithm for the
\textsc{Bipartite Biclique Edge Cover}. Finally,
we present experimental results where we
applied some of our techniques to a real-life data set.\\
{\bf Keywords:} Computational biology, biclique edge covering, bipartite graph,
boolean matrix, NP-completeness, graph theory, fixed-parameter tractability, kernelization.
\end{abstract}

\section{Introduction}
\label{Section: Introduction}
%%%%%%%%%%%%%%%%%%%%%%%%%%%%%%%%%%%%%%%%%%%%%%%%%%%%%%%%%%%%%%%%%%
%%%%%%% Section: Introduction
%%%%%%%%%%%%%%%%%%%%%%%%%%%%%%%%%%%%%%%%%%%%%%%%%%%%%%%%%%%%%%%%%%

{\it Wolbachia} is a genus of inherited, intra-cellular bacteria that
infect many arthropod species, including a significant proportion of
insects.
The bacterium was first identified in 1924 by M. Hertig and
S. B. Wolbach in Culex pipiens, a species of mosquito.
{\it Wolbachia} spreads by altering the reproductive capabilities of
its hosts \cite{Engelstadter2009}. One of these alterations consists
in inducing so-called \emph{cytoplasmic incompatibility} \cite{Engelstadter2009b}.
This phenomenon, in its simplest expression, results in the death of
embryos produced in crosses between males carrying the infection and
uninfected females. A more complex pattern is the death of embryos
seen in crosses between males and females carrying different {\it
  Wolbachia} strains. The study of {\it Wolbachia} and cytoplasmic incompatibility is of
interest due to the high incidence of such infections, amongst others
in human disease vectors such as mosquitoes, where cytoplasmic
incompatibility could potentially be used as a driver mechanism for
the genetic manipulation of natural populations.

The molecular mechanisms underlying cytoplasmic incompatibility are currently unknown, but the observations are consistent with a ``toxin / antitoxin'' model~\cite{Poinsot2003}.
According to this model, the bacteria present in males modify the
sperm (the so-called modification, or mod factor) by depositing a
``toxin'' during its maturation. Bacteria present in females, on the
other hand, deposit an antitoxin (rescue, or resc factor) in the eggs,
so that offsprings of infected females can develop normally. The
simple compatibility patterns seen in several insect hosts
species~\cite{Bordenstein2007,Mercot2004,Dobson2001} has lead to the
general view that cytoplasmic incompatibility relies on a single pair
of mod / resc genes. However, more complex patterns, such as those
seen in Table~\ref{t:cdata} of the mosquito {\it Culex
  pipiens}~\cite{Duron2006}, suggest that this conclusion cannot be
generalized. The aim of this paper is to provide a first model and
algorithm to determine the minimum number of mod and resc genes
required to explain a compatibility dataset for a given insect
host. Such an algorithm will have an important impact on the understanding
of the genetic architecture of cytoplasmic incompatibility. Beyond
{\it Wolbachia}, the method proposed here can be applied to any
parasitic bacteria inducing cytoplasmic incompatibility.

\begin{figure}
\center
\includegraphics[width=14cm,height=8cm]{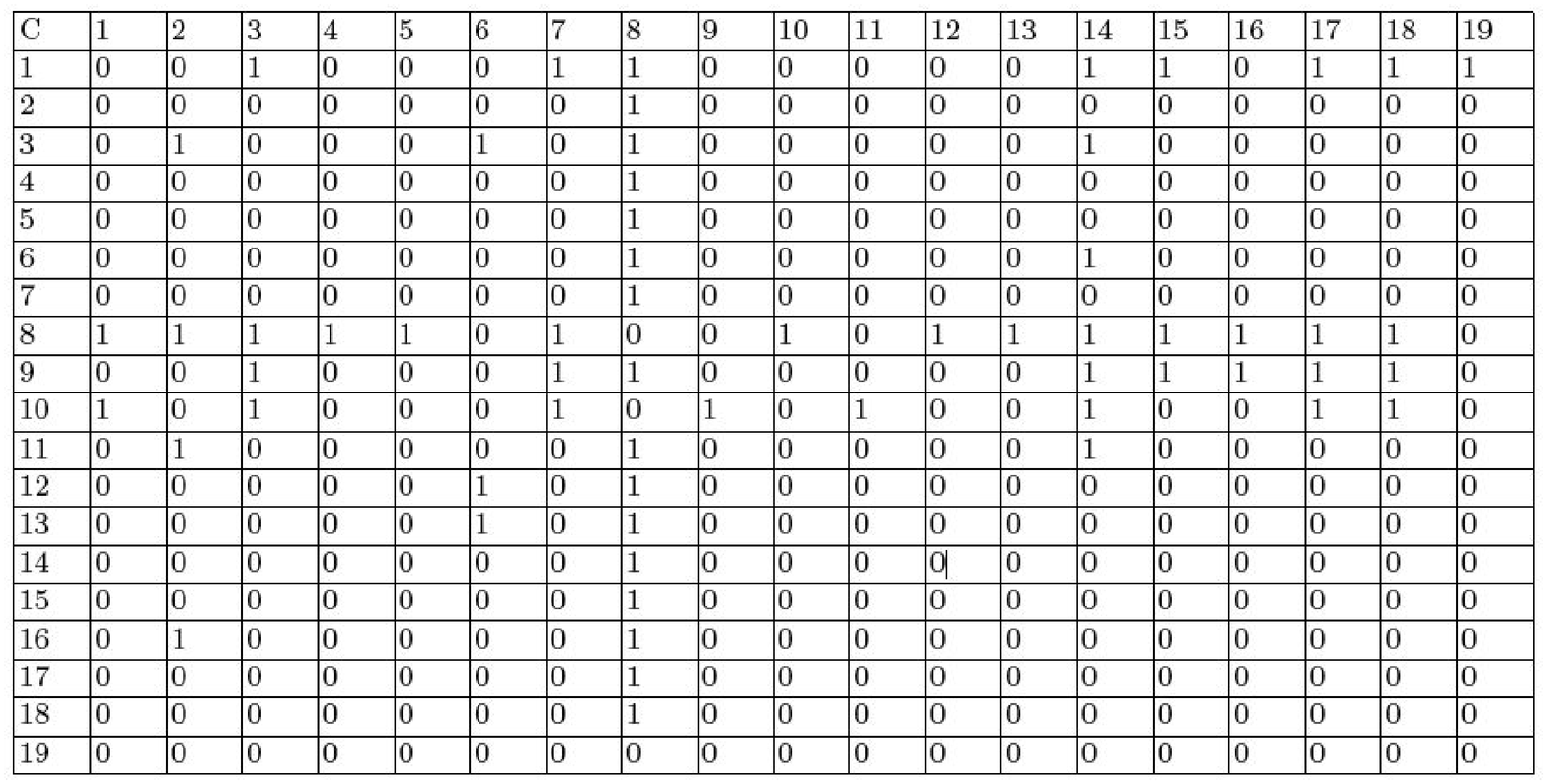}
\caption{The {\it Culex pipiens} dataset. Rows represent females and
  columns males.}
\label{t:cdata}
\end{figure}

% IN: Enlarge the size at the end

\vspace{-0.4cm}
Let us now propose a formal description of this problem. Let the \emph{compatibility
matrix} $C$ be an $n$-by-$n$ matrix describing the observed
cytoplasmic compatibility relationships among $n$ {\it Wolbachia} strains, with
females in rows and males in columns. For the {\it Culex pipiens}
dataset, the content of the $C$ matrix is directly given by
Table~\ref{t:cdata}. For each entry $C_{i,j}$ of this matrix, a value of $1$
indicates that the cross between the $i$'th female and $j$'th
male is incompatible, while a value of $0$ indicates it is
compatible. No intermediate levels of incompatibility are
observed in {\it Culex pipiens}, so that such a discrete code (0 or
1) is sufficient to describe the data. Let the \emph{mod
matrix} $M$ be an $n$-by-$k$ matrix, with $n$ strains and $k$ 
mod genes. For each $M_{i,j}$ entry, a $0$ indicates that strain
$i$ does not carry gene $j$, and a $1$ indicates that it does carry
this gene. Similarly, the \emph{rescue matrix} $R$ is an
$n$-by-$k$ matrix, with $n$ strains and $k$ resc genes, where 
$R_{i,j}$ entries indicate whether strain $i$ carries gene $j$.
A cross between female $i$ and male $j$ is compatible only if
strain $i$ carries at least all the rescue genes matching the
mod genes present in strain $j$. Using this rule, one can
assess whether an $(M,R)$ pair is a solution to the $C$ matrix,
that is, to the observed data.

We can easily find non-parsimonious solutions to this problem,
that is, large $M$ and $R$ matrices that are solutions to $C$, as will
be proven in the next section. However, solutions may also exist with
fewer mod and resc genes.
We are interested in the minimum number of
genes for which solutions to $C$ exist,
and the set of solutions for this minimum number. This problem
can be summarized as follows: Let $C$ (compatibility) be a
boolean $n$-by-$n$ matrix. A pair of $n$-by-$k$ boolean 
matrices $M$ (mod) and $R$ (resc) is called a solution to $C$
if, for any row $j$ in $R$ and row $i$ in $M$, $C_{i,j} = 0$ if 
and only if $R_{j,\ell} \geq M_{i,\ell}$ holds 
for all $\ell$, $1 \leq \ell \leq k$. 
This appropriately models the fact stated above that,
for any cross to be compatible, the
female must carry at least all the rescue genes matching the
mod genes present in the male.
For a given matrix $C$, we are
interested in the minimum value of $k$ for which solutions to
$C$ exist, and the set of solutions for this minimum $k$. We
refer to this problem as the
\textsc{Mod/Resc Parsimony Inference}
problem (see also Section~\ref{Section: Problem Definition}).
Since in come cases, data (on females or males) may be missing, the
compatibility matrix $C$ has dimension $n$-by-$m$ for $n$ not
necessarily equal to $m$. We will consider this more general situation in
what follows.

In this paper, we present the \textsc{Mod/Resc Parsimony Inference} problem and prove it
is equivalent to a well-studied graph-theoretic problem known in
the literature by the name of \textsc{Bipartite Biclique Edge Cover}. In this problem, we are given a bipartite graph, and we want to cover its edges with a minimum number of complete
bipartite subgraphs (bicliques). This problem is known to be
NP-complete, and thus \textsc{Mod/Resc Parsimony Inference} turns out to
be NP-complete as well. In Section~\ref{Section: Fixed-parameter tractability},
we investigate a previous fixed-parameter tractability approach~\cite{FleischnerMujuniPaulusmaSzeider2009} for solving the
\textsc{Bipartite Biclique Edge Cover} problem
and improve its algorithm. In addition, we show a reduction between this
problem and the \textsc{Clique Edge Cover} problem. Finally, in
Section~\ref{Section: Experimental Results}, we present
experimental results where we applied some of these techniques
to the {\it Culex pipiens} data set presented in
Table~\ref{t:cdata}. This provided a surprising finding from a biological point of view.

\section{Problem Definition and Notation}
\label{Section: Problem Definition}
%%%%%%%%%%%%%%%%%%%%%%%%%%%%%%%%%%%%%%%%%%%%%%%%%%%%%%%%%%%%%%%%%%
%%%%%%% Section: Problem Definition and Notation
%%%%%%%%%%%%%%%%%%%%%%%%%%%%%%%%%%%%%%%%%%%%%%%%%%%%%%%%%%%%%%%%%%

In this section, we briefly review some notation and
terminology that will be used throughout the paper. We also
give a precise mathematical definition of the \textsc{Mod/Resc
Parsimony Inference} problem we study. For this, we first need
to define a basic operation between two boolean vectors:

\begin{definition}
The
$\otimes$ vectors multiplication is an operation between two
boolean vectors $U,V \in \{0,1\}^k$ such that~:
$$
U \otimes V :=\left\{
\begin{array}{lll}
1 \quad & : & \quad U[i] > V[i] \textrm{ for some } i \in \{1,\ldots,k\}\\
0 \quad & : & \quad \textrm{otherwise}
\end{array}
\right.
$$
In other words, the result of the $\otimes$ multiplication is
$0$ if, for all corresponding locations, the value in the second
vector is not less than in the first.
\end{definition}

The reader should note that this operation is not symmetric.
For example, if $U:=(0,1,1,0)$ and $V:=(1,1,1,0)$, then $U
\otimes V = 0$,  while $V \otimes U = 1$. We next generalize
the $\otimes$ multiplication to boolean matrices.
This follows easily from the observation that the boolean vectors $U,V \in
\{0,1\}^k$ may be seen as matrices of dimension $1$-by-$k$. We thus
use the same symbol $\otimes$ to denote the operation applied to matrices.
\begin{definition}
The $\otimes$ row-by-row matrix multiplication is a function $\{0,1\}^{n
\times k} \times \{0,1\}^{m \times k} \to \{0,1\}^{n \times m}$
such that $C=M \otimes R$ iff $C_{i,j}=M_i \otimes R_j$ for all
$i \in \{1,\ldots,n\}$ and $j \in \{1,\ldots m\}$. (Here $M_i$
and $R_j$ respectively denote the $i$'th and $j$'th row of $M$
and $R$.)
\end{definition}

\begin{definition}
In the \textsc{Mod/Resc Parsimony Inference} problem, the input
is a boolean matrix $C \in \{0,1\}^{n \times m}$, and the goal
is to find two boolean matrices $M\in \{0,1\}^{n \times k}$ and
$R\in \{0,1\}^{m \times k}$ such that $C_{i,j}=M \otimes R$ and
with $k$ minimal.
\end{definition}
 
We first need to prove there is always a correct solution to the \textsc{Mod/Resc Inference Problem}. Here we show that there is always a solution for as many mod and resc genes as the minimum between the number of male and female strains in the dataset.

\begin{lemma}\label{l:ldef}
The \textsc{Mod/Resc Parsimony Inference} problem always has a
solution.
\end{lemma}

\begin{proof}
A satisfying output for the \textsc{Mod/Resc Parsimony Inference}
problem always exists for any possible $C$ of size $n$-by-$m$.
For instance, let $M$ be of size $n$-by-$n$ and equal to the identity matrix, and
let $R$ be of size $m$-by-$n$ and such that
$R=\overline{C}^T$.
This solution is correct since the only $1$-value in an arbitrary row
$r_i$ of the matrix $M$ is at location $M_{ii}$. Thus, the only situation
where $C_{ij}=1$ is when $R_{ji}=0$, which is the case by construction.
\qed
\end{proof}

We will be using some standard graph-theoretic terminology and
notation. We use $G$, $G'$, and so forth to denote graphs in general,
where $V(G)$ denotes the vertex set of a graph $G$, and $E(G)$
its edge-set.
By a \emph{subgraph} of $G$, we mean a graph $G'$
with $V(G') \subseteq V(G)$ and $E(G') \subseteq E(G)$. For a
bipartite graph $G$, \emph{i.e.} a graph whose vertex-set can
be partitioned into two classes with no edges occurring between
vertices of the same class, we use $V_1(G)$ and $V_2(G)$ to
denote the two vertex classes of $G$. A \emph{complete
bipartite graph} (\emph{biclique}) is a bipartite graph $G$
with $E(G) := \{\{u,v\} : u \in V_1(G),v \in V_2(G) \}$. We
will sometimes use $B$, $B_1$, and so forth to denote
bicliques.

\section{Equivalence to Bipartite Biclique Edge Cover}
\label{Section: Equivalence}
%%%%%%%%%%%%%%%%%%%%%%%%%%%%%%%%%%%%%%%%%%%%%%%%%%%%%%%%%%%%%%%%%%
%%%%%%% Section: Equivalence to Bipartite Biclique Edge Cover
%%%%%%%%%%%%%%%%%%%%%%%%%%%%%%%%%%%%%%%%%%%%%%%%%%%%%%%%%%%%%%%%%%

In this section, we show that the \textsc{Mod/Resc Parsimony
Inference} problem is equivalent to a classical and well-studied
graph theoretical problem known in the literature as the
\textsc{Bipartite Graph Biclique Edge Cover} problem. Using
this equivalence, we first derive the complexity status of
\textsc{Mod/Resc Parsimony Inference}, and later devise FPT
algorithms for this problem. We begin with a formal definition
of the \textsc{Bipartite Graph Biclique Edge Cover} problem.

\begin{definition}
In the \textsc{Bipartite Biclique Edge Cover Problem} problem,
the input is a bipartite graph $G$, and the goal is to find the
minimum number of biclique subgraphs $B_1,\ldots,B_k$ of $G$
such that $E(G):=\bigcup_\ell E(B_\ell)$.
\end{definition}

Given a bipartite graph $G$ with $V_1(G):=\{u_1,\ldots,u_n\}$
and $V_2(G):=\{u_1,\ldots,u_m\}$, the \emph{bi-adjacency} matrix
of $G$ is a boolean matrix $A(G) \in \{0,1\}^{n \times m}$
defined by $A(G)_{i,j} := 1 \iff \{u_i,v_j\} \in E(G)$. In this
way, every boolean matrix $C$ corresponds to a bipartite graph,
and vice versa.

\begin{theorem}
\label{Theorem: Equivalence}%
Let $C$ be a boolean matrix of size $n \times m$. Then there
are two matrices $M\in \{0,1\}^{n \times k}$ and $R\in
\{0,1\}^{m \times k}$ with $C=M \otimes R$ iff the bipartite
graph $G$ with $A(G):=C$ has a biclique edge cover with $k$
bicliques.
\end{theorem}

\begin{proof}
\bf $(\Longleftarrow)$ \rm Let $G$
be the bipartite graph with the bi-adjacency matrix $C$, and
suppose $G$ has biclique edge cover $B_1, B_2, \ldots, B_k$. We
construct two boolean matrices $M$ and $R$ as follows: Let
$V_1(G):=\{u_1,\ldots,u_n\}$ and $V_2(G):=\{v_1,\ldots,v_m\}$.
We define:
\begin{enumerate}
\item $M_{i,\ell} = 1 \iff u_i \in V_1(B_\ell)$.
\item $R_{j,\ell} = 0 \iff v_j \in V_2(B_\ell)$.
\end{enumerate}
An illustration of this construction is given in
Figure~\ref{f:redex}.

\begin{figure}
\center
\includegraphics[width=14cm,height=9cm]{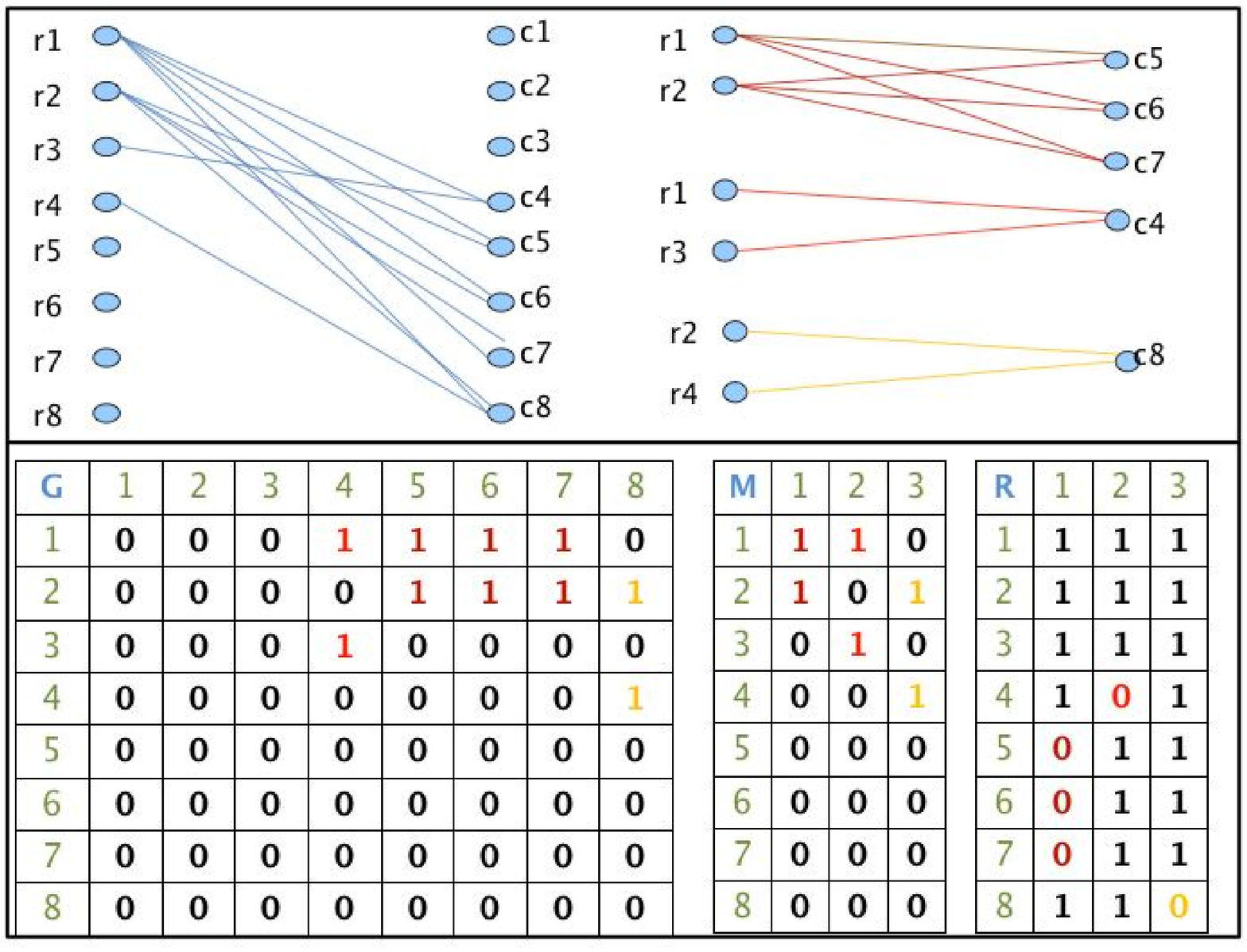}
\caption{Reduction illustrated.} \label{f:redex}
\end{figure}

We argue that $C=M \otimes R$. Consider an arbitrary location
$C_{i,j}=1$. By definition we have $\{u_i,v_j\} \in E(G)$.
Since the bicliques $B_1,\ldots,B_k$ cover all edges of $G$, we
know that there is some $\ell$, $\ell \in \{1,\ldots,k\}$, with
$u_i \in V_1(B_\ell)$ and $v_j \in V_2(B_\ell)$. By 
construction we know that $M_{i,\ell}=1$ and $R_{j,\ell}=0$,
and so $M_i \otimes R_j = 1$, which means that the entry at row
$i$ and column $j$ in $M \otimes C$ is equal to 1. On the other
hand, if $C_{ij}=0$, then $\{u_i,v_j\} \notin E(G)$, and thus 
there is no biclique $B_\ell$ with $u_i \in V_1(B_\ell)$ and 
$v_j \in V_2(B_\ell)$. As a result, for all $\ell \in
\{1,\ldots,k\}$, if $M_{i,\ell} = 1$ then $R_{i,\ell} = 1$ as
well, which means that the result of the $\otimes$
multiplication between the $i$'th row in $M$ and the $j$'th row in
$R$ will be equal to 0.

\bf $(\Longrightarrow)$ \rm Assume there are two matrices $M \in
\{0,1\}^{n \times k}$ and $R \in \{0,1\}^{m \times k}$ with
$C=M \otimes R$. Construct $k$ subgraphs $B_1,\ldots,B_k$ of
$G$, where the $\ell$'th subgraph is defined as follows:
\begin{enumerate}
\item $u_i \in V_1(B_\ell) \iff M_{i,\ell}=1$.
\item $v_j \in V_2(B_\ell) \iff R_{j,\ell}=0$.
\item $\{u_i,v_j\} \in E(B_\ell) \iff \{v_i,v_j\} \in
    E(G)$.
\end{enumerate}

We first argue that each of the subgraphs $B_1,\ldots,B_k$ is a
biclique. Consider an arbitrary subgraph $B_\ell$, and an
arbitrary pair of vertices $u_i \in V_1(B_\ell)$ and $v_j 
V_2(B_\ell)$. By construction, it follows that $M_{i,\ell}=1$ 
and $R_{i,\ell}=0$. As a result, it must be that $C_{i,j}=1$,
which means that $\{u_i,v_j\} \in E(G)$. Next, we argue that
$\bigcup_\ell E(B_\ell) = E(G)$. Consider an arbitrary edge
$\{u_i,v_j\} \in E(G)$. Since $C=A(G)$, we have $C_{i,j}=1$.
Furthermore, since $M \otimes R = C$, there must be some $\ell
\in \{1,\ldots,k\}$ with $M_{i,\ell} > R_{j,\ell}$. However,
this is exactly the condition for having $u_i$ and $v_j$ in the 
biclique subgraph $B_\ell$. It follows that indeed
$\bigcup_\ell E(B_\ell) = E(G)$, and thus the theorem is
proved. \qed
\end{proof}

Due to the equivalence between \textsc{Mod/Resc Parsimony
Inference} and \textsc{Bipartite Biclique Edge Cover}, we can
infer from known complexity results regarding \textsc{Bipartite
Biclique Edge Cover} the complexity of our problem. First, since
\textsc{Bipartite Biclique Edge Cover} is well-known to be
NP-complete~\cite{Orlin1977}, it follows that \textsc{Mod/Resc Parsimony 
Inference} is NP-complete as well. Furthermore, Gruber and
Holzer~\cite{GruberHolzer2007} recently showed that
\textsc{Bipartite Biclique Edge Cover} problem cannot be
approximated within a factor of $n^{1/3-\varepsilon}$ unless
P = NP where $n$ is the total number of vertices. Since the reduction given in Theorem~\ref{Theorem:
Equivalence} is clearly an approximate preserving reduction, we
can deduce the following:

\begin{theorem}
\label{Theorem: Inapproximabilty}%
\textsc{Mod/Resc Parsimony Inference} is
\textnormal{NP-complete}, and furthermore, for all $\varepsilon
> 0$, the problem cannot be approximated within a factor of
$(n+m)^{1/3-\varepsilon}$ unless \textnormal{P} =
\textnormal{NP}.
\end{theorem}

\section{Fixed-parameter tractability}
\label{Section: Fixed-parameter tractability}
%%%%%%%%%%%%%%%%%%%%%%%%%%%%%%%%%%%%%%%%%%%%%%%%%%%%%%%%%%%%%%%%%%
%%%%%%% Section: Fixed-parameter tractability
%%%%%%%%%%%%%%%%%%%%%%%%%%%%%%%%%%%%%%%%%%%%%%%%%%%%%%%%%%%%%%%%%%

In this section, we explore a parameterized complexity approach
\cite{DowneyFellows1999,FlumGrohe2006,Niedermeier2006} for the
\textsc{Mod/Resc Parsimony Inference} problem. Due to the equivalence shown in the previous section, we focus
for convenience reasons on \textsc{Bipartite Biclique Edge
Cover}. In parameterized
complexity, problem instances are appended with an additional
parameter, usually denoted by $k$, and the goal is to find an
algorithm for the given problem which runs
in time $f(k) \cdot n^{O(1)}$, where $f$ is an arbitrary
computable function. In our context, our goal is to determine
whether a given input bipartite graph $G$ with $n$ vertices
has a biclique edge cover of size $k$ in time $f(k) \cdot
n^{O(1)}$.

\subsection{The kernelization}\label{s:sskern}

Fleischner \emph{et
al.}~\cite{FleischnerMujuniPaulusmaSzeider2009} studied the
\textsc{Bipartite Biclique Edge Cover} problem in the context
of parameterized complexity. The main result in their paper is to
provide a kernel for the problem based on the techniques given
by Gramm \emph{et al.}~\cite{GrammGuoHuffnerNiedermeier2006}
for the similar \textsc{Clique Edge Cover} problem.
Kernelization is a central technique in parameterized
complexity which is best described as a polynomial-time
transformation that converts instances of arbitrary size to
instances of a size bounded by the problem parameter (usually of
the same problem), while mapping ``yes''-instances to
``yes''-instances, and ``no''-instances to ``no''-instances. More
precisely, a \emph{kernelization algorithm} $\mathcal{A}$ for a
parameterized problem (language) $\Pi$ is a polynomial-time
algorithm such that there exists some computable function
$f$, such that, given an instance $(I,k)$ of $\Pi$, $\mathcal{A}$
produces an instance $(I',k')$ of $\Pi$ with:
\begin{itemize}
\item $|I'| + k' \leq f(k)$, and
\item $(I,k) \in \Pi \iff (I',k') \in \Pi$.
\end{itemize}

We refer the reader to
\emph{e.g.}~\cite{GuoNiedermeier2007,Niedermeier2006} for more
information on kernelization.

A typical kernelization algorithm works with reduction rules,
which transform a given instance to a slightly smaller
equivalent instance in polynomial time. The typical argument
used when working with reduction rules is that once none of
these can be applied, the resultant instance has size bounded
by a function of the parameter. For the \textsc{Bipartite
Biclique Edge Cover}, two kernelization rules have been applied
by Fleischner \emph{et
al.}~\cite{FleischnerMujuniPaulusmaSzeider2009}:\\

\noindent \underline{\textsf{RULE 1}}: If $G$ has a vertex with
no neighbors, remove this vertex without changing the
parameter.

\noindent \underline{\textsf{RULE 2}}: If $G$ has two vertices
with identical neighbors, remove one of these vertices without
changing the parameter.

\begin{lemma}[\cite{FleischnerMujuniPaulusmaSzeider2009}]
\label{Lemma: Kernel}%
Applying rules 1 and 2 of above exhaustively gives a
kernelization algorithm for \textsc{Bipartite Biclique Edge
Cover} that runs in $O(n^3)$ time, and transforms an instance
$(G,k)$ to an equivalent instance $(G',k)$ with $|V(G')| \leq
2^k$ and $|E(G')| \leq 2^{2k}$.
\end{lemma}

We add two additional rules, which will be necessary
for further interesting properties.

\noindent \underline{\textsf{RULE 3}}: If there is a vertex $v$
with exactly one neighbor $u$ in $G$, then remove both $v$ and
$u$, and decrease the parameter by one.

\begin{lemma}\label{lemma: lrul3}
Rule 3 is correct.
\end{lemma}

\begin{proof}
Assume a biclique cover of size $k$ of the graph, and assume
that vertex $v$ is a member of some of the bicliques in this
cover. By definition, at least one of the bicliques covers the
edge 
$\{u,v\}$. Since this is the only edge 
adjacent to $v$, the bicliques that cover
$\{u,v\}$ include only vertex $u$ among the vertices in its
bipartite vertex class.
If the bicliques do not cover all the
edges of $u$, add them to each of the bicliques. \qed
\end{proof}

\noindent \underline{\textsf{RULE 4}}: If there is a vertex $v$
in $G$ which is adjacent to all vertices in the opposite
bipartition class of $G$, then remove $v$ without decreasing
the parameter.

\begin{lemma}\label{Lemma: lrul4}
Rule 4 is correct.
\end{lemma}

\begin{proof}
After applying rule 3 above, each remaining vertex in the graph
has at least two neighbors. Assume a biclique cover of size $k$
of all the edges except those adjacent to vertex $v$.
Assume w.l.o.g. that $v \in V_1(G)$. Since each vertex
$u \in V_2(G)$ has degree at least 2, it is adjacent to an edge
which is covered by the biclique cover. It therefore belongs to
some biclique in this cover. For each biclique in the cover,
add now vertex $v$ to its set of vertices. Since $v$ is
adjacent to all the vertices of 
$V_2(G)$, each changed component is a correct biclique and the
new solution covers all the edges, including those of vertex
$v$, and is of same size. \qed
\end{proof}

Regarding the time complexity of the new rules we introduced,
it is clear that once a vertex has been found in which a rule
should be applied, applying each rule takes $O(n)$ time. Thus,
including the time necessary to find such a vertex, the time
required for each rule is $O(n)$. 
Since one can apply the
reduction rules at most $O(n)$ time, the total time required
for our extended kernelization remains $O(n^3)$. We remark that
although 
the new rules do not change the kernelization size, which
remains $2^k$ vertices in a solution of size $k$, they will be
useful in the following section.

\subsection{\textsc{Bipartite Biclique Edge Cover} and \textsc{Clique Edge Cover}}\label{s:ssrbeccec}

In this section, we show the connection between the \textsc{Bipartite Biclique
Edge Cover} and the \textsc{Clique Edge Cover} problems. We show that in the context of
fixed-parameter tractability, we can easily translate our problem to
the classical clique covering problem and then use it for a solution to our problem. For instance, it gives another way for the kernelization of the problem and can provide interesting heuristics, mentioned in~\cite{GrammGuoHuffnerNiedermeier2006}. 

Given a kernelized bipartite graph $G'$ as an instance to the
\textsc{Bipartite Biclique Edge Cover} problem, we transform
$G'$ into a (non-bipartite) graph $G''$ defined by $V(G''):=V(G')$
and $E(G''):=E(G') \cup \{\{u,v\} : u,v \in V_1(G') \textrm{ or }
u,v \in V_2(G')\}$.

\begin{theorem}\label{t:qmqh}
The edges of $G'$ can be covered with $k$ cliques iff the edges
of $G''$ can be covered with $k+2$ cliques.
\end{theorem}

\begin{proof}
Suppose $B_1,\ldots,B_k$ is a biclique edge cover of $G'$. Then
each $V(B_i)$, $i \in \{1,\ldots,k\}$, induces a clique in
$G''$. Furthermore, the only remaining edges which are not
covered in $G''$ are the ones between vertices in $V_1(G')$ and
$V_2(G')$, which can be covered by the two cliques induced by
these vertex sets in $G''$. Altogether this gives us $k+2$
cliques that cover all edges in $G''$. Conversely, take a clique
edge cover $K_1,\ldots,K_c$ of $G''$. Due to the fourth
kernelization rule, we know that there is no vertex in $V_1(G')$
which is connected to all vertices in $V_2(G')$, and vice-versa,
in both $G'$ and $G''$. It follows that there must be at least
two cliques in $\{K_1,\ldots,K_c\}$, say $K_1$ and $K_2$, with
$V(K_1) \subseteq V_1(G')$ and $V(K_2) \subseteq V_2(G')$. Thus,
there is a subset of the cliques in $\{K_3,\ldots,K_c\}$ which
have vertices in both partition classes of $G'$, and which cover
all the edges in $G'$. Taking the corresponding bicliques in
$G'$, and adding duplicated bicliques if necessary, gives us $k$ bicliques
that cover all edges in $G'$. \qed
\end{proof}

\subsection{Algorithms}\label{s:sssa}

After the kernelization algorithm is applied, the next step
is usually to solve the problem using brute-force. This is what is done
in~\cite{FleischnerMujuniPaulusmaSzeider2009}. However, the
time complexity given there is inaccurate, and the
parametric-dependent time bound of their algorithm is
$O(k^{4^k}2^{3k})=O(2^{2^{2k}\lg k+3k})$ instead of the $O(2^{2k^2+3k})$ bound 
stated in their paper. Furthermore, the algorithm they describe is
initially given for the related \textsc{Bipartite Biclique Edge
  Partition} problem (where each edge is allowed to appear exactly once in a
biclique), and the adaptation of such algorithm to the \textsc{Bipartite Biclique
Edge Cover} problem is left vague and imprecise.
Here, we suggest two possible brute-force procedures for the \textsc{Bipartite Biclique
Edge Cover} problem, each of which outperforms the algorithm
of~\cite{FleischnerMujuniPaulusmaSzeider2009} in the worst-case. We assume
throughout that we are working
with a kernelized instance obtained by applying the algorithm described in Section~\ref{s:sskern}, \emph{i.e.} a pair $(G',k)$ where $G'$ is a bipartite graph with
at most $2^k$ vertices (and consequently at most $4^k$ edges).

\paragraph{The first brute-force algorithm:}
For each $k' \leq k$, try all possible partitions of the
edge-set $E(G')$ of $G'$ into $k'$ subsets. For each such
partition $\Pi =\{E_1,\ldots,E_{k'}\}$, check whether each of the
subgraphs $G'[E_1],\ldots,G'[E_{k'}]$ is a biclique, where
$G'[E_i]$ is the subgraph of $G$ induced by $E_i$. If yes, report
$G'[E_1],\ldots,G'[E_{k'}]$ as a solution. If some $G'[E_i]$ is
not a biclique, check whether edges in $E(G') \setminus E(G'_i)$
can be added to $E[G'_i]$ in order to make the graph a biclique.
Continue with the next partition if some graph in
$G'[E_1],\ldots,G'[E_{k'}]$ cannot be appended in this way in
order to get a biclique, and otherwise report the solution
found. Finally, if the above procedure fails for all partitions
of $E(G')$ into $k' \leq k$ subsets, report that $G'$ does not have a
biclique edge cover of size $k$.

\begin{lemma}
\label{Lemma: Correctness}
The above algorithm correctly determines whether $G'$ has a
bipartite biclique edge cover of size~$k$ in time
$\frac{2^{2^{2k}\lg k+2k+\lg k}}{k!}$.
\end{lemma}

\begin{proof}
Correctness of the above algorithm is immediate in case a
solution is found. To see that the algorithm is also correct
when it reports that no solution can be found, observe that for
any biclique edge cover $B_1,\ldots,B_k$ of $G$, the set
$\{E_1,\ldots,E_k\}$ with $E_i := E(G'_i) \setminus \bigcup_{j <
i} E(G'_j)$ defines a partition of $E(G')$ (with some of the $E_i$'s
possibly empty), and given this partition, the algorithm above
would find the biclique edge cover of $G'$. Correctness of the
algorithm thus follows.

Regarding the time complexity, the time needed for appending
edges to each subgraph is at most $O(|(V(G'))^2|)=O(2^{2k})$, and thus
a total of $O(2^{2k}k)=O(2^{2k+\lg k})$ time is required for the entire partition. The 
number of possible partitions of $E(G')$ into $k$ disjoint set
is the \emph{Stirling number of the second kind} $S(2^{2k},k)$,
which has been shown in~\cite{Korshunov1983} to be
asymptotically equal to
$O(\frac{k^{4^k}}{k!} = \frac{2^{2^{2k}\lg k}}{k!})$. Thus, the 
total complexity of the algorithm is $\frac{2^{2^{2k}\lg k+2k+\lg k}}{k!}$. 
\qed
\end{proof}

\paragraph{The second brute-force algorithm:}
We generate
the set $\mathcal{K}(G')$ of all possible inclusion-wise maximal
bicliques in $G'$, and try all possible $k$-subsets of
$\mathcal{K}(G')$ to see whether one covers all edges in $G'$.
Correctness of the algorithm is immediate
since one can always restrict oneself to using only
inclusion-wise maximal bicliques in a biclique
edge cover. To generate all maximal
bicliques, we first transform $G'$ into the graph $G''$ given in
Theorem~\ref{t:qmqh}. Thus, every
inclusion-wise maximal biclique in $G'$ is an inclusion-wise
maximal clique in $G''$. We then use the algorithm
of~\cite{TsukiyamaIdeAriyoshiShirakawa1977} on the complement
graph $\overline{G''}$ of $G''$, \emph{i.e.} the graph defined
by $V(\overline{G''}):=V(G'')$ and $E(\overline{G''}) : =
\{\{u,v\} : u,v \in V(\overline{G''}), u \neq v, \textrm{ and }
\{u,v\} \notin E(G'') \}$.

\begin{theorem}
The \textsc{Bipartite Biclique Edge Cover} problem can be
solved in $O(f(k) + n^3)$ time, where $f(k):=2^{k2^{k-1}+3k}$.
\end{theorem}

\begin{proof}
Given a bipartite graph $G$ as an instance to \textsc{Bipartite
Biclique Edge Cover}, we first apply the kernelization algorithm
to obtain an equivalent graph $G'$ with $2^k$ vertices, and then apply the brute-force algorithm
described above to determine whether $G'$ has a biclique edge
cover of size $k$. Correctness of this algorithm follows
directly from Section~\ref{s:sskern} and the correctness of
the brute-force procedure. To analyze the time complexity of
this algorithm, we first note that Prisner showed that any
bipartite graph on $n$ vertices has at most $2^{n/2}$
inclusion-wise maximal
bicliques~\cite{TsukiyamaIdeAriyoshiShirakawa1977}. This
implies that $|\mathcal{K}(G')|  \leq 2^{2^{k-1}}$. The
algorithm of~\cite{Prisner2000} runs in
$O(|V(G')||E(G')||\mathcal{K}(G')|)$ time, which is
$O(2^k2^{2k}2^{2^{k-1}})=O(2^{2^{k-1}+3k})$. Finally, the total
number of $k$-subsets of $\mathcal{K}(G')$ is
$O(2^{k2^{k-1}})$, and checking whether each of these subsets
covers the edges of $G'$ requires $O(|V(G')||E(G')|)=O(2^{3k})$
time. Thus, the total time complexity of the entire algorithm
is $O(2^{2^{k-1}+3k} + 2^{k2^{k-1}+3k} + n^3) =
O(2^{k2^{k-1}+3k} + n^3).$ \qed
\end{proof}

It is worthwhile mentioning that some particular bipartite
graphs have a number of inclusion-wise maximal bicliques, which
is polynomial in the number of their vertices. For these types
of bipartite graphs, we could improve on the worst-case analysis
given in the theorem above. For instance, a bipartite chordal
graph $G$ has at most $|E(G)|$ inclusion-wise maximal
bicliques~\cite{TsukiyamaIdeAriyoshiShirakawa1977}. A bipartite
graph with $n$ vertices and no induced cocktail-party graph of
order $\ell$ has at
most $n^{2(\ell-1)}$ inclusion-wise maximal
bicliques~\cite{Prisner2000}. The cocktail party graph of order $\ell$
is the graph with nodes consisting of two rows of paired nodes in which all nodes but the paired ones are connected with a graph edge
(for a full definition, see~\cite{Prisner2000}). Observing that the algorithm in
Section~\ref{s:sskern} preserves cordiality and does not
introduce any new cocktail-party induced subgraphs, we obtain
the following corollary:

\begin{corollary}
The \textsc{Bipartite Biclique Edge Cover} problem can be
solved in $O(2^{2k^2+3k} + n^3)$ time when restricted to
chordal bipartite graphs, and in $O(2^{2k^2(\ell-1)+3k} + n^3)$
time when restricted to bipartite graphs with no induced
cocktail-party graphs of order $\ell$.
\end{corollary}

\section{Experimental Results}
\label{Section: Experimental Results}
%%%%%%%%%%%%%%%%%%%%%%%%%%%%%%%%%%%%%%%%%%%%%%%%%%%%%%%%%%%%%%%%%%
%%%%%%% Section: Experimental Results
%%%%%%%%%%%%%%%%%%%%%%%%%%%%%%%%%%%%%%%%%%%%%%%%%%%%%%%%%%%%%%%%%%

We performed experiments of the parameterized algorithms on the {\it
Culex pipiens} dataset, given in Table~\ref{t:cdata}. We implemented
the algorithms in the C++ programming language, with source code of
approximately 2500 lines.

%MF: commented out the old version; new one below
%The main difficulty in practice is to prove there is no
%solution of a certain size. Thus our program first checks if there is no
%solution of small sizes, since this is both easy to check using the $FPT$
%approach
% and important for the biology meaning as will be explain
%below. Then,
%MF: which ones?
%using different heuristics we try to discover fast
%and efficient solutions of size $k$. Afterwards, we try to
%improve this by finding a solution of size $k-1$ using the $FPT$
%approach and repeat this process until the program does not finding a
%solution for sufficient number of time. The source code and the
%results can be viewed on the webpage {\em
%http://lbbe.univ-lyon1.fr/-Nor-Igor-.html}.
%MF: you need to say also how long the algorithm takes on the Culex dataset!

%MF: new! Check

The main difficulty in practice is to find the minimal size
$k$. Different approaches could be used. One would proceed by
first checking if there is no
solution of small sizes since this is easy to check using the $FPT$
approach, and then increasing the size
until reaching a smallest size $k$ for which one
solution exists. Another would proceed by using
different fast and efficient heuristics to discover
a solution of a given size $k'$ that in general will be greater than
the optimal size $k$ sought. Then applying dichotomy (the optimal
solution is between 1 and $k'-1$), the minimal size could be found
using the $FPT$ approach for the middle value between 1 and $k'-1$,
and so on.
The source code and the
results can be viewed on the webpage {\em
http://lbbe.univ-lyon1.fr/-Nor-Igor-.html}.

%IN: new: the website. Need to add the results and the code there.

The result obtained on the {\it Culex pipiens}
dataset indicates that $8$ pairs of mod/resc genes are required to explain the
dataset. This appear to be in sharp contrast to more simple patterns
seen in other host
species~\cite{Mercot2004,Dobson2001,Bordenstein2007}
that had led to the general belief that cytoplasmic incompatibility
can be explained with a single pair of mod / resc genes. In biological
terms, this result means that contrary to earlier beliefs, the number
of genetic determinants of cytoplasmic incompatibility present in a
single {\it Wolbachia} strain can be large, consistent with the view that it
might involve repeated genetic elements such as transposable elements
or phages. \\\\\\

\bibliographystyle{plain}
\bibliography{paper}

\begin{thebibliography}{10}

\bibitem{Bordenstein2007}
S.R. Bordenstein and J.H. Werren.
\newblock Bidirectional incompatibility among divergent {\it wolbachia} and
  incompatibility level differences among closely related {\it wolbachia} in
  {\it nasonia}.
\newblock {\em Heredity}, Sep(99(3)):278--87, 2007.

\bibitem{Mercot2004}
H.~Mer\c cot and S.~Charlat.
\newblock {\it Wolbachia} infections in {\it drosophila melanogaster} and {\it
  d. simulans}: polymorphism and levels of cytoplasmic incompatibility.
\newblock {\em Genetica}, 120(1-3):51--9, 2004 Mar.

\bibitem{Dobson2001}
S.L. Dobson, E.J. Marsland, and W.~Rattanadechakul.
\newblock {\it Wolbachia}-induced cytoplasmic incompatibility in single- and
  superinfected {\it aedes albopictus} (diptera: {\it Culicidae}).
\newblock {\em J Med Entomol.}, May(38(3):382--7, 2001.

\bibitem{DowneyFellows1999}
R.G. Downey and M.R. Fellows.
\newblock {\em Parameterized Complexity}.
\newblock Springer-Verlag, 1999.

\bibitem{Duron2006}
O.~Duron, C.~Bernard, S.~Unal, A.~Berthomieu, C.~Berticat, and M.~Weill.
\newblock Tracking factors modulating cytoplasmic incompatibilities in the
  mosquito {\it culex pipiens}.
\newblock {\em Mol Ecol.}, Sep(15(10)):3061--3071, 2006.

\bibitem{Engelstadter2009}
J.~Engelstadter and G.D.D. Hurst.
\newblock The ecology and evolution of microbes that manipulate host
  reproduction.
\newblock {\em Annual Review of Ecology, Evolution and Systematics},
  (40):127--149, 2009.

\bibitem{Engelstadter2009b}
J.~Engelstadter and A.~Telschow.
\newblock Cytoplasmic incompatibility and host population structure.
\newblock {\em Heredity}, (103):196--207, 2009.

\bibitem{FleischnerMujuniPaulusmaSzeider2009}
H.~Fleischner, E.~Mujuni, D.~Paulusma, and S.~Szeider.
\newblock Covering graphs with few complete bipartite subgraphs.
\newblock {\em Theoretical Computer Science}, 410(21-23):2045--2053, 2009.

\bibitem{FlumGrohe2006}
J.~Flum and M.~Grohe.
\newblock {\em Parameterized Complexity Theory}.
\newblock Springer, 2006.

\bibitem{GrammGuoHuffnerNiedermeier2006}
J.~Gramm, J.~Guo, F.~Huffner, and R.~Niedermeier.
\newblock Data reduction, exact, and heuristic algorithms for clique cover.
\newblock In {\em Proceedings of the 8th ACM/SIAM workshop on ALgorithm
  ENgineering and EXperiments (ALENEX)}, pages 86--94, 2006.

\bibitem{GruberHolzer2007}
H.~Gruber and M.~Holzer.
\newblock Inapproximability of nondeterministic state and transition complexity
  assuming {P}$\neq${NP}.
\newblock In {\em Proceedings of the 11th international conference on
  Developments in Language Theory (DLT)}, pages 205--216, 2007.

\bibitem{GuoNiedermeier2007}
J.~Guo and R.~Niedermeier.
\newblock Invitation to data reduction and problem kernelization.
\newblock {\em SIGACT News}, 38(1):31--45, 2007.

\bibitem{Korshunov1983}
A.D. Korshunov.
\newblock Asymptotic behaviour of stirling numbers of the second kind.
\newblock {\em Diskret. Anal.}, 39(1):24–41, 1983.

\bibitem{Niedermeier2006}
R.~Niedermeier.
\newblock {\em Invitation to Fixed-Parameter Algorithms}.
\newblock Oxford University Press, 2006.

\bibitem{Orlin1977}
J.~Orlin.
\newblock Contentment in graph theory: covering graphs with cliques.
\newblock {\em Indagationes Mathematicae}, 80(5):406–424, 1977.

\bibitem{Poinsot2003}
D.~Poinsot, S.~Charlat, and H.~Mer\c cot.
\newblock On the mechanism of {\it wolbachia}-induced cytoplasmic
  incompatibility: confronting the models with the facts.
\newblock {\em Bioessays}, 25(1):259--265, 2003.

\bibitem{Prisner2000}
E.~Prisner.
\newblock Bicliques in graphs {I}: Bounds on their number.
\newblock {\em Combinatorica}, 20(1):109--117, 2000.

\bibitem{TsukiyamaIdeAriyoshiShirakawa1977}
S.~Tsukiyama, M.~Ide, H.~Ariyoshi, and I.~Shirakawa.
\newblock A new algorithm for generating all the maximal independent sets.
\newblock {\em SIAM Journal on Computing}, 6(3):505--517, 1977.

\end{thebibliography}

\section{Appendix}
\label{Section: Appendix}
%%%%%%%%%%%%%%%%%%%%%%%%%%%%%%%%%%%%%%%%%%%%%%%%%%%%%%%%%%%%%%%%%%
%%%%%%% Section: Appendix
%%%%%%%%%%%%%%%%%%%%%%%%%%%%%%%%%%%%%%%%%%%%%%%%%%%%%%%%%%%%%%%%%%

\begin{figure}
\center
\includegraphics[width=14cm,height=10cm]{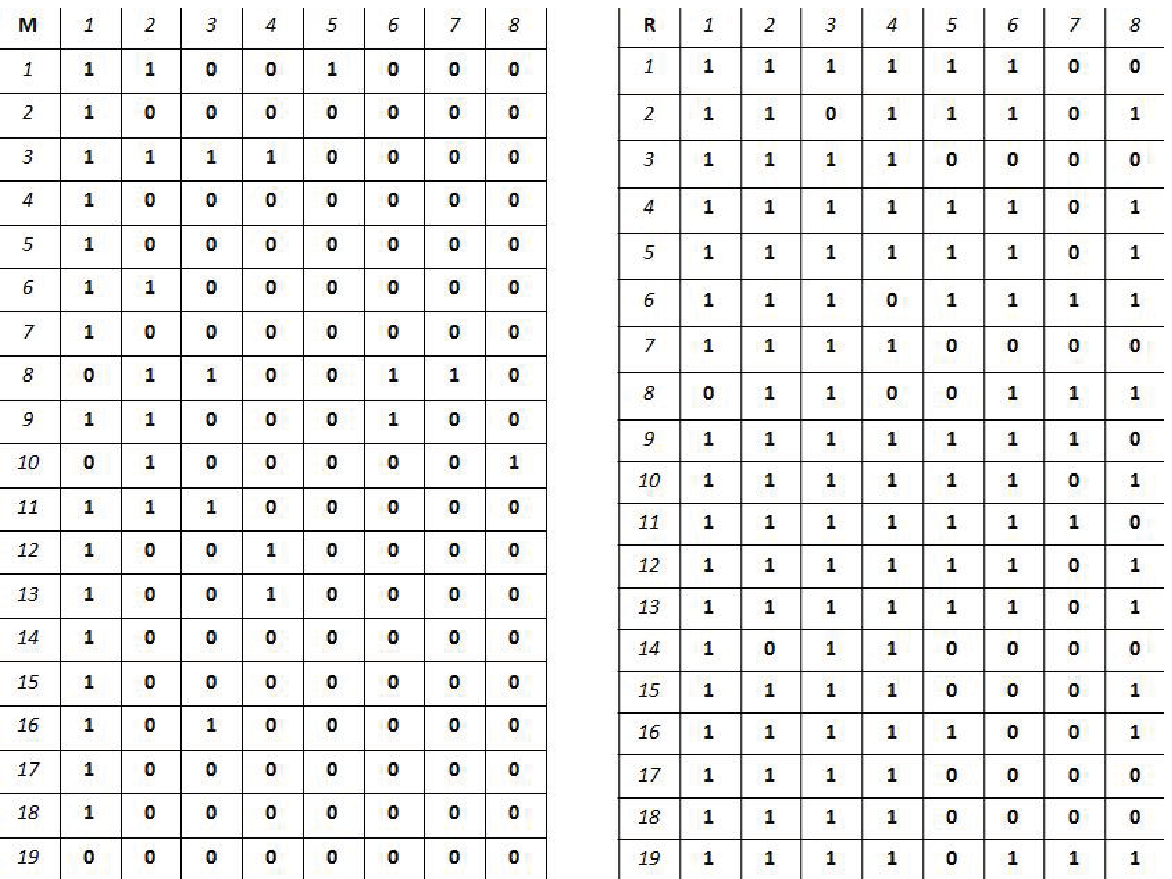}
\caption{The {\it Culex pipiens} solution.}
\label{t:outdata}
\end{figure}

\end{document}